\documentclass[12pt]{article}

\usepackage{amsmath}
\usepackage{amssymb}
\usepackage{amsfonts}
\usepackage{latexsym}
\usepackage{braket}
\usepackage{color}

\catcode `\@=11 \@addtoreset{equation}{section}

\catcode `\@=12

\newcommand{\be}{\begin{equation}}
	\newcommand{\en}{\end{equation}}
\newcommand{\bea}{\begin{eqnarray}}
	\newcommand{\ena}{\end{eqnarray}}
\newcommand{\beano}{\begin{eqnarray*}}
	\newcommand{\enano}{\end{eqnarray*}}
\newcommand{\bee}{\begin{enumerate}}
	\newcommand{\ene}{\end{enumerate}}

\newcommand{\Sc}{{\cal S}}

\newcommand{\F}{{\cal F}}
\newcommand{\G}{{\cal G}}

\newcommand{\1}{1 \!\! 1}

\newtheorem{thm}{Theorem}

\newtheorem{lemma}[thm]{Lemma}
\newtheorem{prop}[thm]{Proposition}

\newenvironment{proof}{\noindent {\bf Proof --}}{\hfill$\square$ \vspace{3mm}\endtrivlist}

\catcode `\@=11 \@addtoreset{equation}{section}
\catcode `\@=12

\begin{document}
	
	\thispagestyle{empty}

	\vspace*{2cm}
	
	\begin{center}
		{\Large \bf On the Berry-Keating Operator\footnote{This paper is dedicated to the memory of Franciszek Hugon Szafraniec. We like to imagine he would have enjoyed reading it.}}   
		
		\vspace{5mm}
		
		{\large Fabio Bagarello}\\
		 Dipartimento di Ingegneria,
		Universit\`a di Palermo,\\ I-90128  Palermo, Italy\\
		and I.N.F.N., Sezione di Catania\\
		e-mail: fabio.bagarello@unipa.it\\

		\vspace{8mm}

		{\large Sergiusz Ku\.{z}el}\\
		AGH University, Krak\'ow, Poland\\
		e-mail: kuzhel@agh.edu.pl\\

	\end{center}
	
	\vspace*{1cm}
	
	\begin{abstract}
	\noindent {We review here two different viewpoints on the Berry–Keating operator $H_{\mathrm{BK}}$, whose connection to the Riemann hypothesis remains an intriguing and not yet fully understood question, despite considerable attention in the recent literature.} 
    In particular, we propose two somehow complementary views to $H_{BK}$: the first is based on a purely Hilbertian point of view, on dilation operators and on the Mellin transform. The second is a distributional approach, with a specific view to ladder operators, generalized eigenstates of $H_{BK}$, and generalized coherent states.
	\end{abstract}
	
	\vspace{0.5cm}
    
{\bf Keywords}: Berry-Keating operator, Mellin transform,  pseudo bosons, bi-coherent states 

\vspace{0.5cm}

{\bf Mathematics Subject Classification}: 47B93, 47B25, 79, 81V73

	
	\newpage
	
	\section{Introduction}
The Berry-Keating operator represents a profound and ambitious attempt to connect the zeros of the Riemann zeta function—a central object in number
theory—to the spectral theory of a quantum Hamiltonian. First proposed by Michael Berry and Jonathan Keating \cite{BK}, this operator quantizes the classical system $H=xp$, producing an operator whose spectrum has been conjectured to encode the non-trivial zeros of the Riemann zeta function. This proposal may be viewed as a modern realization of the Hilbert–Pólya conjecture—associated with David Hilbert and George Pólya—which interprets the Riemann hypothesis as the assertion that these zeros arise as eigenvalues of a suitable self-adjont operator. Specifically, the imaginary parts of the zeros located on the critical line $\text{Re}(s) = \frac{1}{2}$ are identified with the eigenvalues of the self-adjoint operator.  This conjecture is supported by numerical and heuristic considerations, in particular by the observed statistical agreement between the distribution of non-trivial zeros of the zeta function and the distribution of eigenvalues of random Hermitian matrices, as predicted by random matrix theory \cite{FM, Mon, OZ}.

 After proper symmetrization, the operator considered in \cite{BK} was 
$$
H_{BK}=\frac{1}{2}\left(\hat x\hat p+\hat p\hat x\right),
$$ where $\hat x$ and $\hat p=-i\frac{d}{dx}$ are the usual self-adjoint position and momentum operators, satisfying (in the sense of unbounded operators) the commutation relation $[\hat x,\hat p]=i\1$, in units in which $\hbar=1$. 
In \cite{BK} the authors find the generalized eigenfunctions of $H_{BK}$ as follows:
\begin{equation}\label{new1}
\psi_E(x)=x^{-1/2+iE}, \qquad E\in\mathbb{R}
\end{equation}
which satisfy the eigenvalue equation 
$$
H_{BK}\psi_E(x)=E\psi_E(x)
$$
It is worth stressing that $\psi_E(x)$ is not square integrable, for real (or real positive) $E$.
This issue highlights a central difficulty in the Berry–Keating conjecture: the construction of a suitable Hilbert space on which 
$H_{BK}$ is self-adjoint with a discrete spectrum. The absence of a natural inner-product structure makes the analysis particularly challenging. Many authors have attempted to analyze and clarify these issues, yet a complete understanding of the underlying mechanisms has not been achieved \cite{bellisard, bbm, S, EY}. 

The purpose of this paper is to draw attention to the still little-known properties of the Berry–Keating operator. The paper consists of two independent parts. In the first part, the Berry–Keating operator $H_{BK}$ is analyzed in the $L_2$-setting, where the modified Mellin transform and the concept of the Lebesgue spectrum play a central role. Here we highlight the relationship between 
$H_{BK}$ and the Phillips symmetric operator \cite{KN}, and discuss the possibility of developing a Lax-Phillips scattering theory for the Berry-Keating squared operator $H_{BK}^2$ \cite{LF}.

{The results of the first part, in which $H_{BK}$ was studied through its generalized eigenfunctions \eqref{new1}, may indicate that the Hilbert–Pólya conjecture cannot be realized within the $L_2$- framework based on the Berry–Keating operator. A possible route forward, inspired by the classical theory of generalized eigenfunction expansions \cite{BER1, BER2}, is to replace the original family of generalized eigenfunctions  of $H_{BK}$ with a more suitable family.
Motivated by this idea, we show in the second part that the Berry–Keating operator admits a natural description in terms of weak pseudo-bosons (WPBs) \cite{bagweak}.}
More precisely, we demonstrate that $H_{BK}$ can be naturally framed within the WPB formalism and is not substantially different from the number-like operator.
A key aspect of this analysis is the role played by non-square-integrable functions. In particular, we construct a set of generalized eigenstates of $H_{BK}$ in the sense of distributions. Finally, we emphasize that, in contrast to the approach proposed in \cite{bbm}, no metric operator is required in our framework. We also focus on coherent states of a very specific kind, the so called {\em weak bi-coherent states}, \cite{bagspringerbook}, and we discuss some of their properties.

\section{The Berry-Keating operator in $L_2$-setting}

Consider a strongly continuous one-parameter unitary group of dilation
\begin{equation}\label{uman1}
U(t)f=e^{-t/2}f(e^{-t}x), \qquad t\in\mathbb{R}
\end{equation}
acting in the Hilbert space $L_2(\mathbb{R})$
\cite[p. 259]{TE}.
The infinitesimal generator $G$ of $U(\cdot)$ can be easily computed  for $ f\in\mathcal{S}(\mathbb{R})$:
$$
Gf=-i\left(x\frac{d}{dx}+\frac{1}{2}\right)f=\frac{1}{2}\left(\hat x\hat p+\hat p\hat x\right)f=H_{BK}f. 
$$
This means that the Berry-Keating operator $H_{BK}$ can be interpreted as the infinitesimal generator of $U(\cdot)$, i.e.,  
$$
U(t)=e^{-iH_{BK}t}=e^{-\frac{t}{2}}e^{-x\frac{d}{dx}t}.
$$  
In summary, the operator $H_{BK}$ with the domain
of the infinitesimal generator
\begin{equation}\label{uman2}
 \mathcal{D}(H_{BK})=\{f\in{L_2(\mathbb{R})} \ : \ \lim_{t\to{0}}\frac{1}{t}(U(t)-I)f \ \mbox{exists in the $L_2$-norm topology}  \}
 \end{equation} 
 is self-adjoint in $L_2(\mathbb{R})$ and its restriction to  $\mathcal{S}(\mathbb{R})\subset\mathcal{D}(H_{BK})$ is essentially self-adjoint.

It follows from \cite{WA, AP} that an equivalent description of \eqref{uman2} is given as 
\begin{equation}\label{uman4}
\mathcal{D}(H_{BK})=\{f\in{L_2(\mathbb{R})} \ : \ f\in AC_{loc}(\mathbb{R}), \quad
xf'(x)\in{L_2(\mathbb{R})} \}.
\end{equation}

In view of \eqref{uman1}, $$
\mathcal{P}U(t)=U(t)\mathcal{P}, \qquad 
\mathcal{X}U(t)=U(t)\mathcal{X}, \qquad t\in\mathbb{R},
$$
where $\mathcal{P}f(x)=f(-x)$ is the space parity operator and $\mathcal{X}f=({\sf sign}\ x)f(x)$. This fact and \eqref{uman2}  mean that
\begin{equation}\label{uman3b}
H_{BK}\mathcal{P}f=\mathcal{P}H_{BK}f, \qquad H_{BK}\mathcal{X}f=\mathcal{X}H_{BK}f, \qquad f\in\mathcal{D}(H_{BK}).
\end{equation}

In view of \eqref{uman3b}, the subspaces
\begin{equation}\label{uman9}
L_2^{\pm}(\mathbb{R})=(I\pm\mathcal{P})L_2(\mathbb{R}), \qquad L_2(\mathbb{R}_\pm)=(I\pm\mathcal{X})L_2(\mathbb{R})
\end{equation} 
reduce the operator $H_{BK}$ and the restrictions of $H_{BK}$ in these subspaces are infinitesimal generators of the corresponding restrictions of the unitary group $U(\cdot)$.  

A self-adjoint operator $H$
 is said to have \emph{the Lebesgue spectrum of multiplicity $n$}
 if its spectrum is absolutely continuous with respect to the Lebesgue measure on $\mathbb{R}$
 and the spectral multiplicity is equal to $n$ at every point of $\mathbb{R}$.
 Equivalently, 
$H$ has the Lebesgue spectrum of multiplicity $n$
 if and only if it is unitarily equivalent to the operator of multiplication by the independent variable on 
 $L_2(\mathbb{R}, \mathbb{C}^n)$.

\begin{prop}\label{prop1}
The Berry-Keating operator $H_{BK}$ has the Lebesgue spectrum of  multiplicity $2$.
\end{prop}
\begin{proof}
The decomposition of the eigenfunctions 
$\psi_E$  (see \eqref{new1}) into even and odd components produces two eigenfunctions of 
$H_{BK}$:
$$
\psi_E^+(x)=|x|^{-1/2+iE}, \qquad \psi_E^-(x)=({\sf sign}\ x)|x|^{-1/2+iE}.  
$$
This allows us to define two modifications of the Mellin transform: 
\begin{equation}\label{new2}
\begin{array}{c}
({\mathcal{M}}^+f)(E)=\frac{1}{2\sqrt{\pi}}\int_{-\infty}^\infty\psi_E^+(x)f(x)dx=\frac{1}{2\sqrt{\pi}}\int_{-\infty}^\infty{f(x)}{|x|}^{-1/2+iE}dx \vspace{5mm} \\
({\mathcal{M}}^-f)(E)=\frac{1}{2\sqrt{\pi}}\int_{-\infty}^\infty\psi_E^-(x)f(x)dx=\frac{1}{2\sqrt{\pi}}\int_{-\infty}^\infty{f(x)}({\sf sign}\ x){|x|}^{-1/2+iE}dx
\end{array}
\end{equation}

The  operators ${\mathcal M}^+ : L_2^{+}(\mathbb{R}) \to L_2(\mathbb{R})$ and ${\mathcal M}^- : L_2^{-}(\mathbb{R}) \to L_2(\mathbb{R})$  are unitary mappings. Moreover, in view of \cite[Section 11.3]{BBO}, 
\begin{equation}\label{uman7}
\mathcal{M}^+(H^+_{BK}f_+)(E)=E({\mathcal{M}^+}f_+)(E), \qquad \mathcal{M}^-(H^-_{BK}f_-)(E)=E({\mathcal{M}^-}f_-)(E),  \qquad E\in\mathbb{R}, 
\end{equation}
where $f_\pm\in{L_2^\pm(\mathbb{R})}$ and $H^{\pm}_{BK}$ are the restrictions of $H_{BK}$ to $L_2^{\pm}(\mathbb{R})$. 

Consider the decomposition 
$L_2(\mathbb{R})=L_2^+(\mathbb{R})\oplus{L}_2^-(\mathbb{R})$ and the operator $\mathcal{M}=\mathcal{M}^+\oplus\mathcal{M}^-$ defined in $L_2(\mathbb{R})$. By construction, $\mathcal{M}$ is a unitary mapping from $L_2(\mathbb{R})$ to $L_2(\mathbb{R}, \mathbb{C}^2)$, where $\mathcal{M}^+$ maps $L_2(\mathbb{R})$ to $L_2(\mathbb{R}, 
\mathbb{C}\oplus{0})$ and $\mathcal{M}^-$ maps $L_2(\mathbb{R})$ to $L_2(\mathbb{R}, 
0\oplus\mathbb{C})$.
 
  Since $H_{BK}$ has the decomposition $H_{BK}=H^-_{BK}\oplus{H^+_{BK}}$ and \eqref{uman7} satisfied, the operator $\mathcal{M}$ transfers $H_{BK}$ to the multiplication operator by independent variable $E$ in $L_2(\mathbb{R}, \mathbb{C}^2)$. The proof is complete.
\end{proof}

It follows from \eqref{new2} and the definition of $\mathcal{M}$ that the operators $\mathcal{P}$ and $\mathcal{X}$ are unitarily equivalent  to the   multiplication by the 
Pauli matrices $\sigma_3=\left[\begin{array}{cc} 1 & 0 \\
0 & -1 \end{array}\right]$ and $\sigma_1=\left[\begin{array}{cc} 0 & 1 \\
1 & 0 \end{array}\right]$ in $L_2(\mathbb{R}, \mathbb{C}^2)$, respectively.  Precisely, 
$$
{\mathcal{M}}\mathcal{P}f=\sigma_3(\mathcal{M}f)(E), \qquad {\mathcal{M}}\mathcal{X}{f}=\sigma_1(\mathcal{M}f)(E). 
$$
Since $L_2(\mathbb{R}_+)=(I+\mathcal{X})L_2(\mathbb{R})$, the operator $\mathcal{M}$ transforms $L_2(\mathbb{R}_+)$ into the subspace 
$$
(\sigma_0+\sigma_1)L_2(\mathbb{R}, \mathbb{C}^2)=\left\{\left[\begin{array}{c}
v(E)\\
v(E) \end{array}\right] : v(E)\in{L_2(\mathbb{R})}\right\} 
$$
of  $L_2(\mathbb{R}, \mathbb{C}^2)$. 
This means that the restriction of the Berry–Keating operator to 
$L_2(\mathbb{R}_+)$ has the Lebesgue spectrum of multiplicity one and $H_{BK}^+$ is unitarily equivalent to the momentum operator $\hat{p}$ 
on  $L_2(\mathbb{R})$. In \cite{TW}, the operator $H_{BK}^+$ was proposed as a natural candidate for a new momentum operator — referred to as hyperbolic momentum — associated with dilations of 
$\hat{x}$ in the construction of quantum mechanics in $\mathbb{R}_+$. 

Similarly, since $L_2^+(\mathbb{R})=(I+\mathcal{P})L_2(\mathbb{R})$, the operator $\mathcal{M}$ maps the subspace of even functions $L_2^+(\mathbb{R})$
 onto the subspace 
$$
(\sigma_0+\sigma_3)L_2(\mathbb{R}, \mathbb{C}^2)=\left\{\left[\begin{array}{c}
v(E)\\
0 \end{array}\right] : v(E)\in{L_2(\mathbb{R})}\right\} 
$$
of  $L_2(\mathbb{R}, \mathbb{C}^2)$. The restriction of $H_{BK}$
 to $L_2^+(\mathbb{R})$ has the Lebesgue spectrum of multiplicity one.

By reasoning analogous to that used above, one finds that the restriction of $H_{BK}$
 to each of the subspaces \eqref{uman9} possesses the Lebesgue spectrum of multiplicity one. Hence, these operators cannot qualify as candidates for the hypothetical Hilbert–Pólya operator. The remaining task is therefore to find a Hilbert space for which $H_{BK}$
admits a discrete spectrum. One natural strategy is to invoke extension-theoretic methods. For example, \cite{ES} considers the Berry–Keating operator on compact quantum graphs. Although the resulting self-adjoint realizations possess a discrete spectrum, their eigenvalues cannot reproduce the nontrivial zeros of the Riemann zeta function (see No-go theorem 15.6 in \cite{ES}).

Independently of the search for the hypothetical Hilbert–Pólya operator, studying the Berry–Keating operator within the framework of extension theory leads to a number of interesting results.
 As an example, let us consider the restriction of 
$H_{BK}$ to the domain
$$
\mathcal{D}'=\{f\in\mathcal{D}(H_{BK}) : \langle (H_{BK}-iI)f,  \gamma_+ \rangle=\langle (H_{BK}-iI)f,  \gamma_- \rangle=0\}, 
$$
where 
$$
 \gamma_+(x)=\left\{\begin{array}{cc} 
0 & |x|\leq{1} \\
|x|^{-3/2} & |x|>1
\end{array}\right., \qquad \gamma_-(x)=\left\{\begin{array}{cc} 
0 & |x|\leq{1} \\
({\sf sign}\ x)|x|^{-3/2} & |x|>1
\end{array}\right.
$$
\begin{lemma}\label{uman634}
The restriction of 
$H_{BK}$ to
$\mathcal{D}'$ determines a closed densely defined symmetric operator 
$S=H_{BK}|_{\mathcal{D}'}$ with defect indices 
$(2,2)$. The defect subspaces of $S$ corresponding to $\pm{i}$ coincide with $
\ker(S^*+iI)=\mbox{span}\{\gamma_+, \gamma_-\}$ and $  \ker(S^*-iI)=\mbox{span}\{\eta_+, \eta_-\}, 
$
where
$$
 \eta_+(x)=\left\{\begin{array}{cc} 
|x|^{1/2} & |x|\leq{1} \\
0 & |x|>1
\end{array}\right., \qquad \eta_-(x)=\left\{\begin{array}{cc} 
({\sf sign}\ x)|x|^{1/2} & |x|\leq{1} \\
0 & |x|>1
\end{array}\right.
$$
\end{lemma}
\begin{proof} In view of \eqref{uman4}, the two-dimensional subspace $\mbox{span}\{\gamma_+, \gamma_-\}$  contains no nonzero vectors from $\mathcal{D}(H_{BK})$. By virtue of \cite[Lemma 5.1]{KN}, the operator $
S=H_{BK}, \ \mathcal{D}(S)={\mathcal{D}'}$ is a closed, densely defined symmetric operator with defect indices $(2,2)$. Moreover, $\ker(S^*+iI)=\mbox{span}\{\gamma_+, \gamma_-\}$

The functions $\gamma_+, \eta_+$ and $\gamma_-, \eta_-$ are even and odd, respectively, by construction. Applying \eqref{new2}, we obtain
$$
(\mathcal{M}\gamma_{\pm})(E)=(\mathcal{M}^\pm\gamma_{\pm})(E)=\left(\frac{1}{\sqrt{\pi}}\int_1^\infty{x}^{-2+iE}dx\right)e_\pm=\frac{i}{\sqrt{\pi}}\frac{e_{\pm}}{E+i},
$$
$$
(\mathcal{M}\eta_{\pm})(E)=(\mathcal{M}^\pm\eta_{\pm})(E)=\left(\frac{1}{\sqrt{\pi}}\int_0^1{x}^{iE}dx\right)e_\pm=-\frac{i}{\sqrt{\pi}}\frac{e_{\pm}}{E-i},
$$
where $e_+=\left[\begin{array}{c}
1 \\
0 \end{array}\right]$ and $e_-=\left[\begin{array}{c}
0 \\
1 \end{array}\right]$.
Taking the Fourier transform,
$$
(Ff)(\lambda)=\frac{1}{\sqrt{2\pi}}\int_{-\infty}^\infty{e^{-i\lambda{E}}}f(E)dE,
$$ 
yields
$$
(F\mathcal{M}\gamma_{\pm})(\lambda)=\sqrt{2}e_\pm\left\{\begin{array}{cc}
e^{-\lambda} & \lambda>0 \\
0 & \lambda<0\end{array}\right., \qquad (F\mathcal{M}\eta_{\pm})(\lambda)=\sqrt{2}e_\pm\left\{\begin{array}{cc}
0 & \lambda>0 \\
e^{\lambda} & \lambda<0\end{array}\right..
$$

The unitary operator 
$F{\mathcal M}$ 
maps $L_2(\mathbb{R})$ to $L_2(\mathbb{R}, \mathbb{C}^2)$ 
 and transforms the Berry–Keating operator $H_{BK}$ into the momentum operator
$i\frac{d}{d\lambda}$, $\mathcal{D}(i\frac{d}{d\lambda})=W_2^1(\mathbb{R}, \mathbb{C}^2)$
 acting in $L_2(\mathbb{R}, \mathbb{C}^2)$. Moreover, for each $f\in\mathcal{D}(S)$, 
 $$
0=\langle (H_{BK}-iI)f, \gamma_\pm\rangle=\langle F\mathcal{M}(H_{BK}-iI)f, F\mathcal{M}\gamma_\pm\rangle=\langle (i\frac{d}{d\lambda}-iI)u(\lambda), F\mathcal{M}\gamma_\pm\rangle= 
$$
$$
=i\sqrt{2}\int_{0}^\infty(u_{\pm}'(\lambda)-u_{\pm}(\lambda))e^{-\lambda}d\lambda=-i\sqrt{2}u_\pm(0), 
 $$
 where $u(\lambda)=\left[\begin{array}{c}
 u_+(\lambda) \\
 u_-(\lambda)
 \end{array}\right]=(F\mathcal{M}f)(\lambda).$
This means that $S$ is unitary equivalent to the symmetric operator 
\begin{equation}\label{uuu14}
i\frac{d}{d\lambda}, \qquad \mathcal{D}(i\frac{d}{d\lambda})=\left\{u\in{W_2^1}(\mathbb{R}, \mathbb{C}^2) : u(0)=\left[\begin{array}{c}
 0 \\
 0
 \end{array}\right]\right\}.
 \end{equation}
 Taking \eqref{uuu14} into account, we obtain, for $f\in\mathcal{D}(S)$, 
 $$
\langle (S+iI)f, \eta_\pm\rangle=\langle (H_{BK}+iI)f, \eta_\pm\rangle=\langle F\mathcal{M}(H_{BK}+iI)f, F\mathcal{M}\eta_\pm\rangle=\langle (i\frac{d}{d\lambda}+iI)u(\lambda), F\mathcal{M}\eta_\pm\rangle= 
$$
$$
=i\sqrt{2}\int_{-\infty}^0(u_{\pm}'(\lambda)+u_{\pm}(\lambda))e^{\lambda}d\lambda=i\sqrt{2}u_\pm(0)=0 
 $$
 since $u$ belongs to $\mathcal{D}(i\frac{d}{d\lambda})$. Therefore, $\ker(S^*-iI)=\mbox{span}\{\eta_+, \eta_-\}$. The proof is completed. 
\end{proof}

\begin{thm}\label{uman54}
 The spectrum $\sigma(H)$ of  a proper extension $H$ of $S$  coincides with one of the following sets: (i) $\sigma(H)=\mathbb{C}$; (ii)  $\sigma(H)=\mathbb{C}_+\cup\mathbb{R}$ or $\sigma(H)=\mathbb{C}_-\cup\mathbb{R}$; (iii)
$\sigma(H)=\mathbb{R}$. In the latter case, $H$ is similar\footnote{operators $A$  and $B$ are  \emph{similar} if $B=T^{-1}AT$, where $T$ is a bounded operator with bounded inverse } to the Berry–Keating operator 
 $H_{BK}$ in $L_2(\mathbb{R})$. If $H$ is a self-adjoint extension, then $H$ is unitarily equivalent to $H_{BK}$.
\end{thm}
\begin{proof} Denote $\mathfrak{M}=\ker(S^*+iI)\dot{+}\ker(S^*-iI)$. 
According to the von Neumann formulas, any proper extension $H$ of $S$ (that is, $S \subseteq H \subseteq S^*$) is uniquely determined by the choice of a subspace $M\subseteq\mathfrak{M}$. Precisely,
$$
\mathcal{D}(H)=\mathcal{D}(S)\dot{+}M \quad \mbox{and} \quad H={S^*}|_{\mathcal{D}(H)}. 
$$

If $\dim{M}\not=2$, then the spectrum of $H$ coincides with $\mathbb{C}$. 
If $\dim M=2$ , then the operator $H$ can be described using the boundary triplet technique and its spectrum can be characterized with the use of an analytic function (the Weyl function or the characteristic function of $S$) \cite{BHS}. In the simplest case, the characteristic function of $S$ is constant. This condition characterizes a special class of symmetric operators with unusual properties, namely the Phillips symmetric operators \cite{KN, KVS}. One of the simplest examples of a simple Phillips symmetric operator is the momentum operator subject to the zero boundary condition \eqref{uuu14}. The operator $S$ also belongs to this class, since, as follows from the proof of Lemma \ref{uman634}, it is unitarily equivalent to \eqref{uuu14}. Consequently, the characteristic function of $S$ is constant as well. This observation was used in \cite{KN} to investigate proper extensions of a simple Phillips symmetric operator. In particular, the statement of the theorem follows from Proposition 4.1 and Theorem 4.2  in \cite{KN}.
\end{proof}

In view of \eqref{uman1} and the definition of Fourier transform $F$, we obtain
$$
FU(t)=U^{-1}(t)F
$$

Since 
 $H_{BK}$ is the infinitesimal generator of unitary group of dilations  $U(t)$, 
  the infinitesimal generator of $U^{-1}(t)$ coincides with $-H_{BK}$.
This means that 
$$
FH_{BK}=-H_{BK}F.
$$
Therefore the 
``squared" Berry-Keating operator 
$$
H_{BK}^2=\left(-i\left(x\frac{d}{dx}+\frac{1}{2}\right)\right)^2=-x^2\frac{d^2}{dx^2}-2x\frac{d}{dx}-\frac{1}{4}
$$
commutes with $F$. Furthermore, $H_{BK}^2$  
is a self-adjoint and positive operator in $L_2(\mathbb{R})$.

It follows from the proof of Lemma \ref{uman634}  that $H_{BK}^2$ is unitary equivalent to the self-adjoint operator $-\frac{d^2}{d\lambda^2}$ defined on $W_2^2(\mathbb{R}, \mathbb{C}^2)$ in $L_2(\mathbb{R}, \mathbb{C}^2)$. This means that the group of unitary operators $e^{-iH_{BK}^2t}$ determines a free evolution in the Lax-Phillips scattering approach \cite{KM}. One should also note that a similar result holds for every restriction of $H_{BK}^2$
 to the subspaces 
\eqref{uman9} of $L_2(\mathbb{R})$. This simple observation opens the way to applying the well-developed framework of the Lax–Phillips approach to the analysis of various perturbations of $H_{BK}^2$.

 The operator $H_{BK}^2$ is connected with the Black–Scholes operator \cite{AP, ES}, whose mathematical properties have been widely investigated, also for its applications in finance, \cite{baaquie, bag2016, roy, roy2}. This makes this connection even more interesting and surely deserves more attention in the future.

\section{ Weak pseudo bosons for the Berry-Keating operator}

In this section we will consider a completely different point of view on $H_{BK}$, showing that this operator can be analyzed in terms of weak pseudo-bosons (WPBs), \cite{bagspringerbook}, and that in fact $H_{BK}$ is not particularly different from the number-like operator. 

Let us consider the operators 
 $$
 a=\frac{d}{dx}=i\hat{p}, \qquad  b=\hat {x}f(x)=xf(x).
 $$  
Since 
$[a,b]f=f$  holds for all $f\in\Sc(\mathbb{R})$,
it is natural to interpret $a$ and $b$
as $\Sc(\mathbb{R})$-pseudo bosons \cite{bagspringerbook}. However, see \cite{bagweak}, if we look for the vacua of $a$ and $b^\dagger=\hat x$, we easily find that $\varphi_0(x)=1$ and $\psi_0(x)=\delta(x)$, with a suitable choice of normalizations. Therefore, it is clear that neither $\varphi_0(x)$ nor $\psi_0(x)$ belongs to  $L_2(\mathbb{R})$. 

Using the ladder properties of $a$ and $b$, we have
\be
\varphi_n(x)=\frac{b^n}{\sqrt{n!}}\,\varphi_0(x)=\frac{x^n}{\sqrt{n!}}, \qquad \psi_n(x)=\frac{(a^\dagger)^n}{\sqrt{n!}}\,\psi_0(x)=\frac{(-1)^n}{\sqrt{n!}}\,\delta^{(n)}(x),
\label{31}\en
for all $n=0,1,2,3,\ldots$. This suggests considering $a^\dagger$ and $b$ as linear operators acting on $\Sc'(\mathbb{R})$. This is possible (and natural) since $a, b, a^\dagger$ and $b^\dagger$ all map $\Sc'(\mathbb{R})$ to itself. For this reason, we can consider the following version of the pseudo-bosonic commutation rule:
$$
[a,b] \varphi (x) = \varphi (x), \qquad \varphi(x)\in\Sc'(\mathbb{R}).
$$
From (\ref{31}) it is clear that $b$ and $a^\dagger$ act as raising operators, respectively, on the sets $\F_{\varphi}=\{\varphi_n\}_{n=0}^\infty$ and $\F_{\psi}=\{\psi_n\}_{n=0}^\infty$:
$$
b\varphi_n=\sqrt{n+1}\varphi_{n+1}, \qquad \qquad a^\dagger\psi_n=\sqrt{n+1}\psi_{n+1}, \qquad n\in\mathbb{N}\cup\{0\}.
$$
 Similarly $b^\dagger$ and $a$ act as lowering operators on these sets:
$$
a\varphi_n=\sqrt{n}\varphi_{n-1}, \qquad \qquad b^\dagger\psi_n=\sqrt{n}\psi_{n-1}, \qquad n\in\mathbb{N}\cup\{0\}
$$
with the understanding that $a\varphi_0(x)=b^\dagger\psi_0(x)=0$.

If we now set $N=ba$ and $N^\dagger=a^\dagger b^\dagger$, we see \cite{bagweak} that
$$
N\varphi_n= n\varphi_n, \qquad N^\dagger \psi_n=n\psi_n, \qquad n\in\mathbb{N}\cup\{0\}
$$
 We should probably stress that both $N$ and $N^\dagger$ are  {\it extensions} of the analogous operators originally defined on a dense subset of $L_2(\mathbb{R})$ to a much larger set, $\Sc'(\mathbb{R})$.  It is possible to see that the various distributions $\varphi_n$ and $\psi_n$ are (generalized) eigenfunctions of $H_{BK}$ and $H_{BK}^\dagger$, with purely imaginary eigenvalues. Indeed, we can rewrite
$$
H_{BK}=-i\left(ba+\frac{1}{2}I\right)=-i\left(N+\frac{1}{2}I\right), \qquad H_{BK}^\dagger=i\left(N^\dagger+\frac{1}{2}I\right).
$$
Hence, we have 
$$
H_{BK}\varphi_n=E_n\varphi_n, \qquad H_{BK}^\dagger\psi_n=\overline E_n\psi_n, 
$$
where $E_n=-i(n+\frac{1}{2})$.

Let us now consider the set of all these generalized eigenvectors, $\F_\varphi=\{\varphi_n\}_{n=0}^\infty$ and $\F_\psi=\{\psi_n\}_{n=0}^\infty$. In \cite{bagweak} the possibility of using  $\F_\varphi$ and $\F_{\psi}$ to expand the appropriate functions has been discussed. In particular, it has been shown that
$$
\sum_{n=0}^\infty \left<\psi_n,f\right>\varphi_n(x)=\sum_{n=0}^\infty \frac{(-1)^n}{n!} \left<\delta^{(n)},f\right>x^n=\sum_{n=0}^\infty \frac{1}{n!} f^{(n)}(0)\,x^n=f(x),
$$
for all $f(x)$ admitting a Taylor expansion. However, if we invert the role of $\F_{\psi}$ and $\F_{\varphi}$, the result is more complicated:
$$
\sum_{n=0}^\infty \left<\varphi_n,f\right>\psi_n(x)=\sum_{n=0}^\infty\frac{(-1)^n}{n!}\left<x^n,f\right>\delta^{(n)}(x),
$$
which is what is sometimes called a {\em dual Taylor series}. It is known that the series does not define in general an element of $\mathcal{D}'(\mathbb{R})$, a distribution (hence it cannot define a tempered distribution) except when the series above returns a finite sum, which may happen only for specific $f$. In this case, it is clear that this sum is indeed a tempered distribution.  

It has been shown in \cite{bagweak} that the sets $\F_\varphi$ and $\F_\psi$ are biorthonormal in a suitable sense. In fact, since neither $\varphi_n$ nor $\psi_n$ belong to $L_2(\mathbb{R})$, we cannot use the usual scalar product on it to check biorthonormality. But, see again \cite{bagweak}, it was shown that it is possible to extend the standard scalar product to our case. Let us summarize the procedure. First, we observe that the scalar product between two suitably well-behaved functions—e.g., $f, g \in \Sc(\mathbb{R})$—can be expressed in terms of a convolution
between $\overline{f(x)}$ and the function $\tilde{g}(x)=g(-x)$: $$\left<f,g\right>=(\overline{f}* \tilde{g})(0).$$ 
Then, we define the scalar product between two elements $F(x), G(x)\in\Sc'(\mathbb{R})$ as the following convolution:
$$
\left<F,G\right>=(\overline{F}* \tilde{G})(0),
$$
whenever this convolution exists.  This is exactly what happens for us, and we get \cite{bagweak},
$$
\left<\varphi_n,\psi_m\right>=\delta_{n,m}.
$$

In \cite{baggarg} it has been proven that $\F_{\varphi}$ and $\F_{\psi}$ are $\Sc_{\cal A}(\mathbb{R})$-quasi bases, where 
$$
	\Sc_{\cal A}(\mathbb{R})=\Sc(\mathbb{R})\cap {\cal A}(\mathbb{R})
$$
	and where $A(\mathbb{R})$ is the set of  entire real analytic functions, which admit expansion in Taylor series, convergent everywhere in $\mathbb{C}$. Of course $\Sc_{\cal A}(\mathbb{R})$ contains many functions of $\Sc(\mathbb{R})$, but not all. We conclude that
$$
\left<f,g\right>=\sum_{n=0}^\infty\,\left<f,\psi_n\right>\left<\varphi_n,g\right>=\sum_{n=0}^\infty\,\left<f,\varphi_n\right>\left<\psi_n,g\right>,
$$
for all $f(x),g(x)\in\Sc_{\cal A}(\mathbb{R})$.

We refer to \cite{bagweak,bagspringerbook} for more properties of the sets $\F_\varphi$ and $\F_{\psi}$. 

It is now possible to introduce the so-called {\em weak bi-coherent states} in connection with the operators $a$ and $b$. This aspect is somehow related to  Professor Szafraniec's scientific interests, who was very much involved with coherent states in many different ways. We refer to \cite{franek2,franek1} for some of his papers on this specific topic. The main difference between what we are briefly going to discuss here, and Franciszek Szafraniec's contributions, is in the role of distributions in the Berry-Keating Hamiltonian.

\subsection{Bi-coherent states}\label{sect3}

Bi-coherent states (BCS) are extensions of coherent states in a similar way as biorthonormal bases are extensions of orthonormal bases. In particular, the resolution of the identity involves both sets of BCS. But, see \cite{bagspringerbook}, this resolution of the identity is recovered in $L_2(\mathbb{R})$, or in some dense subspace of it. Here, in view of the nature of the vectors in (\ref{31}), $L_2(\mathbb{R})$ does not appear to be the natural space to work with. In fact, we expect that a more relevant aspect should be played by  functional spaces used in distributions theory. This is indeed what we will briefly review here, focusing on some results first deduced in \cite{baggarg}.

We start introducing the set
$$
\G=\left\{f(x)\in\Sc(\mathbb{R}): \, e^{kx}f(x)\in\Sc(\mathbb{R}), \, \forall k\in\mathbb{C}\right\}.
$$
This set is dense in $L_2(\mathbb{R})$, since it contains $D(\mathbb{R})$, the set of compactly supported $C^\infty$ functions. In \cite{baggarg} it has been proven that, $\forall f(x)\in\G$, the series $\sum_{k=0}^\infty\,\frac{z^k}{\sqrt{k!}}\,\langle f,\varphi_k\rangle$ converges to $\int_{\mathbb{R}}\overline{f(x)} e^{zx}\,dx$ for all $z\in\mathbb{C}$. 

This suggests to define a functional $F_\varphi$ on $\G$ as follows:
\be
F_\varphi[f](z,\overline z)=e^{-\frac{|z|^2}{2}}\sum_{k=0}^\infty\,\frac{\overline z^k}{\sqrt{k!}}\,\langle \varphi_k,f\rangle=e^{-\frac{|z|^2}{2}}\int_{\mathbb{R}} e^{\overline z x}\,f(x)\,dx,
\label{42}\en
together with the function $\varphi(z;x)=e^{-\frac{|z|^2}{2}}e^{zx}$, so that we can also write
$$
F_\varphi[f](z,\overline z)=\int_{\mathbb{R}} \overline{\varphi(z;x)}\,f(x)\,dx=\langle \varphi(z),f\rangle.
$$

\vspace{2mm}

The same analysis, but with some difference, can be repeated for the series $\sum_{k=0}^\infty\,\frac{z^k}{\sqrt{k!}}\,\langle f,\psi_k\rangle$. We start recalling that, see \cite{baggarg}, $\langle g,\psi_n\rangle=
\frac{1}{\sqrt{n!}}\overline{g^{(n)}(0)}$ for all $g(x)\in\Sc(\mathbb{R})$. Then
$$
\sum_{k=0}^\infty\,\frac{z^k}{\sqrt{k!}}\,\langle f,\psi_k\rangle=\sum_{k=0}^\infty\,\frac{z^k}{k!}\,\overline{f^{(k)}(0)}=\overline{\sum_{k=0}^\infty\,\frac{\overline z^k}{k!}\,f^{(k)}(0)}=\overline{f(\overline z)}
$$
for all functions in $\Sc_{\cal A}(\mathbb{R})$.
 Here $f(z)=\sum_{k=0}^\infty\,\frac{z^k}{k!}\,{f^{(k)}(0)}$, which is clearly convergent $\forall z\in\mathbb{C}$, because of the definition of ${ A}(\mathbb{R})$. As in (\ref{42}), we define a linear functional $F_\psi$ on $\Sc_{\cal A}(\mathbb{R})$ and its related {\em representation} $\psi(z;x)$, as follows
$$
F_\psi[g](z,\overline z)=e^{-\frac{|z|^2}{2}}\sum_{k=0}^\infty\,\frac{\overline z^k}{\sqrt{k!}}\,\langle \psi_k,g\rangle=e^{-\frac{|z|^2}{2}}g(\overline z)=\int_{\mathbb{R}} \overline{\psi(z;x)}\,g(x)\,dx=\langle \psi(z),g\rangle,
$$
$\forall g(x)\in\Sc_{\cal A}(\mathbb{R})$. This could be formally rewritten as
$$
\psi(z;x)=e^{-\frac{|z|^2}{2}}\delta(x- z),
$$
where the Dirac delta distribution with complex argument appears, see e.g. \cite{complexdelta1}.

Using now the same steps as for ordinary CS we can check that $\varphi(z;x)$ and $\psi(z;x)$ satisfy the following resolution of the identity:
$$
\langle f,g\rangle=\int_{\mathbb{C}}\frac{d^2z}{\pi}\langle f,\varphi(z)\rangle\langle \psi(z),g\rangle=\int_{\mathbb{C}}\frac{d^2z}{\pi}\langle f,\psi(z)\rangle\langle \varphi(z),g\rangle,
$$
for all $f(x),g(x)\in\Sc_{\cal A}(\mathbb{R})\cap\G$. Indeed we have, for instance
$$
\int_{\mathbb{C}}\frac{d^2z}{\pi}\langle f,\varphi(z)\rangle\langle \psi(z),g\rangle=\int_{\mathbb{C}}\frac{d^2z}{\pi}e^{-|z|^2}\left(\sum_{k=0}^\infty\,\frac{z^k}{\sqrt{k!}}\,\langle f,\varphi_k\rangle\right)\left(\sum_{l=0}^\infty\,\frac{\overline z^l}{\sqrt{l!}}\,\langle \psi_l,g\rangle\right)=
$$
$$
=\frac{1}{\pi}\sum_{k,l=0}^\infty\frac{\langle f,\varphi_k\rangle\langle \psi_l,g\rangle}{\sqrt{k!\,l!}}\int_{\mathbb{C}}\,d^2z e^{-|z|^2} z^k\,\overline z^l=\langle f,g\rangle,
$$
since $\int_{\mathbb{C}}\,d^2z e^{-|z|^2} z^k\,\overline z^l=\pi\delta_{l,k}k!$. The conclusion follows from the fact that $(\F_\varphi,\F_{\psi})$ are $\Sc_{\cal A}(\mathbb{R})$-quasi bases, as discussed before. 

\vspace{2mm}

{\bf Remark:--} In \cite{baggarg} it has been shown how, playing a little bit with the above equalities, we can deduce an  integral representation of the Dirac delta distribution with complex argument. Indeed we have, using in particular the expressions $\varphi(z;x)=e^{-\frac{|z|^2}{2}}e^{zx}$ and $\psi(z;x)=e^{-\frac{|z|^2}{2}}\delta(x- z)$,
$$
\langle f,g\rangle=\int_{\mathbb{R}}\overline{f(x)}\,g(x)\,dx=\int_{\mathbb{C}}\frac{d^2z}{\pi}\langle f,\varphi(z)\rangle\langle \psi(z),g\rangle=
$$
$$
=\int_{\mathbb{C}}\frac{d^2z}{\pi}\left(e^{-\frac{|z|^2}{2}}\int_{\mathbb{R}}\overline{f(x)}\,e^{zx}\,dx\right)\left(e^{-\frac{|z|^2}{2}}g(\overline z)\right)=\int_{\mathbb{R}}\overline{f(x)}\left[\int_{\mathbb{C}}\frac{d^2z}{\pi}e^{-|z|^2}e^{zx}g(\overline z)\right]\,dx,
$$
 changing the order of the integration.

This equality should be satisfied for all  $f(x),g(x)\in\Sc_{\cal A}(\mathbb{R})\cap\G$, which is true if the following equality holds, at least weakly on $\Sc_{\cal A}(\mathbb{R})\cap\G$:
$$
\int_{\mathbb{C}}\frac{d^2z}{\pi}e^{-|z|^2}e^{zx}g(\overline z)=g(x),
$$
$\forall g(x)\in\Sc_{\cal A}(\mathbb{R})\cap\G$. This identity is exactly the integral representation of the Dirac delta distribution with complex argument we mentioned above. More results on this specific aspect can be found in \cite{baggarg}.

\vspace{2mm}

The vectors $\varphi(z;x)$ and $\psi(z;x)$ are also (weak) eigenstates of the lowering operators $a$ and $b^\dagger$, as expected. Indeed we can prove that $\forall f(x), g(x)\in\Sc_{\cal A}(\mathbb{R})\cap\G$, we have
$$
\langle f,a\varphi(z)\rangle=z\langle f,\varphi(z)\rangle, \qquad \langle g,b^\dagger\psi(z)\rangle=z\langle g,\psi(z)\rangle,
$$
which are our weak versions of the eigenvalue equations for CS. The proof of these equalities is given in \cite{baggarg}.

The conclusion is that, at the price of working weakly on suitable subsets of $L_2(\mathbb{R})$, for the operators $\hat x $ and $i\hat p$ it is possible to introduce two functionals $F_\varphi$ and $F_\psi$, or equivalently two $z$-dependent vectors $\varphi(z;x)$ and $\psi(z;x)$, which share with ordinary CS some of their essential properties. It is important to notice that, in what we have done so far, our weak bi-coherent states have been introduced by two convergent series. We refer to \cite{baggarg} for the role of two  displacement-like operators in connection with our states. This enrich the parallel between our weak BCS and the standard coherent states.

\section{Conclusions}

In this paper, we have reviewed some mathematical features connected with  the Berry-Keating Hamiltonian, both in a Hilbertian and in a distributional framework. The analysis of this operator is extremely reach and far from being concluded. A lot more can be done, and in particular its connection with the Hilbert–Pólya conjecture needs still to be better understood. Particularly amazing, we believe, is the apparent simplicity of $H_{BK}$, which would suggest a much simpler analysis of its mathematical features and physical applications, which is far from being what we are seeing.

	\section*{Acknowledgments}  
S.K. was partially supported by the statutory research tasks of the Faculty of Applied Mathematics at AGH University of Krakow. F. B.  acknowledges partial  support from Palermo University and from G.N.F.M. of the INdAM. {The authors are grateful to the anonymous referee for the careful reading of the manuscript and for helpful suggestions that improved its presentation.} 

\section*{Declarations}
\begin{itemize}
\item {\bf Funding} This research did not receive funding.
\item {\bf Conflict of interest} The authors declare that there is no conflict of interest regarding the publication of
this paper.
\item {\bf Data availability} No datasets were generated or analyzed during the current study.
\item {\bf Author contribution} The paper is equally contributed by both authors.
\end{itemize}


\begin{thebibliography}{99}
    \bibitem{WA} W. Arendt, \emph{Gaussian estimates and interpolation of the spectrum in $L_2$}, Differential and Integral Equations,   {\bf 7} (1994), 1153 - 1168.

\bibitem{AP} W. Arendt,  B. de Pagter, {\it Spectrum and asymptotics of the Black-Scholes partial differential equation in $(L^1, L^\infty)-$interpolation spaces}, Pacific J. Math. {\bf 202} (2002), 1-36. 

\bibitem{baaquie} B.E. Baaquie, \emph{Quantum Finance}, Cambridge University
Press, (2004)

\bibitem{bag2016} F. Bagarello, {\em Appearances of pseudo-bosons from Black-Scholes equation},  J. Math. Phys., {\bf 57}, 043504 (2016)


\bibitem{bagweak} F. Bagarello,  {\em
	Weak pseudo-bosons},   J. Phys. A, {\bf 53}, 135201 (2020)

\bibitem{bagspringerbook} F. Bagarello, {\em Pseudo-Bosons and Their Coherent States}, Springer, Mathematical Physics Studies, (2022)

\bibitem{baggarg} F. Bagarello, F. Gargano, {\em Bi-coherent states as generalized eigenstates of the position and the momentum operators},
ZAMP, {\bf 73} (2022), 119.

\bibitem{BHS} 
J. Behrndt, S. Hassi, H. de Snoo, \emph{Boundary Value Problems, Weyl Functions, and Differential Operators}, Birkh\"auser 2020.

\bibitem{bellisard} J. V. Bellisard, {\em Comment on "Hamiltonian for the zeros of the Riemann zeta function"}, arXiv: 1704.02644, quant-ph (2017)

\bibitem{bbm} C. M. Bender, D. C. Brodje, M. P. M\"uller, {\em Hamiltonian for the zeros of the Riemann zeta function}, Phys. Rev. Lett. {\bf 118} (2017), 130201.

\bibitem{BER1}
Ju. M. Berezanskii, {\em Expansions in 
Eigenfunctions of Selfadjoint Operators,} AMS, 1968. 

\bibitem{BER2}
Ju. M. Berezanskii, 
Z. G. Sheftel, G. F. Us, {\em Functional Analysis},    Birkhäuser, 1996.

\bibitem{BK} M. V. Berry, J. P. Keating, {\em $H=xp$ and the Riemann zeros}, in {\em Supersymmetry and Trace Formulae: Chaos and Disorder}, edited by I.V. Lerner et al., Kluwer Academic/Plenum: New York, 1999.

\bibitem{BBO} Bertrand, J., Bertrand, P., Ovarlez, J. \emph{The Mellin Transform.}
The Transforms and Applications Handbook: Second Edition.Ed. Alexander D. Poularikas
Boca Raton: CRC Press LLC, 2000.

\bibitem{complexdelta1}  R. A. Brewster, J. D. Franson, {\em Generalized delta functions and their use in quantum optics}, J. Math. Phys., {\bf 59}  (2018), 012102.

\bibitem{FM} P. J. Forrester,  A. Mays, {\it Finite-size corrections in random matrix theory and Odlyzko’s dataset for the Riemann zeros}, Proc. R. Soc. A {\bf 471} (2015) 20150436.

\bibitem{ES} S. Endres,   F. Steiner, {\em The Berry-Keating operator on $L_2(\mathbb{R}_+, dx)$ and on compact quantum graphs with general self-adjoint realizations} J. Phys. A: Math. Theor. {\bf 43} (2010), 095204.

\bibitem{franek2} J.-P. Gazeau, F. H. Szafraniec, {\em Holomorphic Hermite polynomials and a non-commutative plane}, J. Phys. A: Math. Theor. {\bf 44} (2011), 49. 

\bibitem{franek1} K. G\'orska, A. Horzela, F. H. Szafraniec, {\em Coherence, squeezing and entanglement -- an example of peaceful coexistence}, in  F. Bagarello,  J.-P. Antoine, J.-P. Gazeau Eds, {\em Coherent states	and applications: a contemporary panorama}, Springer Proceedings in Physics, (2018)

\bibitem{roy} T. K. Jana, P. Roy, {\em Supersymmetry in option pricing}, Phys. A, {\bf 390} (2011), 2350-2355.

\bibitem{roy2} T. K. Jana, P. Roy, {\em Pseudo Hermitian formulation of the quantum Black-Scholes Hamiltonian}, Phys. A, {\bf 391}  (2012), 2636-2640.

\bibitem{KM} S. Kuzhel, 
U. Moskalyova, {\it The Lax-Phillips scattering approach and singular perturbations of Schr\"{o}dinger operator homogeneous with respect to scaling transformations,} J. Math. Kyoto University, {\bf 45} (2005), 265-286.

\bibitem{KN} S. Kuzhel,  L. Nizhnik, \emph{Phillps symmetric operators and their extensions,} Banach J. Math. Anal. {\bf 12} (2018), 995-1016.

\bibitem{KVS} S. Kuzhel, O. Shapovalova,  L. Vavrykovych, \emph{On $J$-self-adjoint extensions of the Phillips symmetric operators}, Meth.  Funct. Anal. Topol.  {\bf 16} (2010),  333-348.

\bibitem{LF} P. D. Lax,  R. F. Phillips, \textit{Scattering Theory}, 2nd ed., with appendices by C. S. Morawetz
and G. Schmidt, Pure Appl. Math. 26, Academic Press, Boston, 1989.

\bibitem{Mon} Hugh L Montgomery, {\it The pair correlation of zeros of the zeta function}, Analytic number theory, Proc. Sympos. Pure Math., vol. XXIV, Providence, R.I.: American Mathematical Society (1973), pp. 181–193.  

\bibitem{OZ}  A.M. Odlyzko,  {\it On the distribution of spacings between zeros of the zeta function}, Mathematics of Computation, {\bf 48} (1987), 273–308.

\bibitem{S} G. Sierra,  {\em The Riemann Zeros as Spectrum and the
    Riemann Hypothesis} Symmetry, {\bf 11} (2019), 494.

 \bibitem{TE} G. Teschl, \emph{Mathematical Methods in Quantum Mechanics
With Applications to Schrodinger Operators,} AMS, 2014.

\bibitem{TW} J. Twamley, G. J. Milburn, \emph{The quantum Mellin transform}, New J. Phys. {\bf 8} (2006), 328.

\bibitem{EY}  E. Yakaboylu, {\em
Hamiltonian for the Hilbert–Pólya conjecture}, J. Phys. A: Math. Theor.
{\bf 57} (2024), 235204.





\end{thebibliography}
\end{document}